\documentclass{article}     

\usepackage{graphicx}
\usepackage{amsmath,amssymb, amsthm}
\usepackage{hyperref}
\usepackage{natbib}
\usepackage{xcolor}
\usepackage{thmtools}

\newcommand{\RR}{\mathbb{R}}

\newcommand{\ria}{\rightarrow}
\newcommand{\set}[1]{\{#1\}}
\newcommand{\NP}{\mathrm{NP}}
\newcommand{\CV}{\mathcal{V}}

\newcommand{\CE}{\mathcal{E}}
\newcommand{\FC}{\mathfrak{C}}
\newcommand{\score}{\mathrm{score}}
\newcommand{\sub}{\subseteq}

\newcommand{\pos}{\mathrm{pos}}

\newcommand{\Interval}{\mathrm{Interval}}
\newcommand{\Dominance}{\mathrm{Dominance}}

\newtheorem{theorem}{Theorem}[section]
\newtheorem{proposition}[theorem]{Proposition}

\declaretheorem[style=definition,sibling=theorem,qed=$\blacksquare$]{definition}
\declaretheorem[style=definition,sibling=theorem,qed=$\blacksquare$,name=Example]{example}

\begin{document}

\title{Electing a committee with dominance constraints\thanks{Support from the Basic Research Program of the National Research University Higher School of Economics is gratefully acknowledged.}
}


\author{Egor Ianovski}


\maketitle

\begin{abstract}
    We consider the problem of electing a committee of $k$ candidates, subject to some constraints as to what this committee is supposed to look like. In our framework, the candidates are given labels as an abstraction of a politician's religion, a film's genre, a song's language, or other attribute, and the election outcome is constrained by  interval constraints -- of the form ``Between 3 and 5 candidates with label X'' -- and dominance constraints -- ``At least as many candidates with label X as with label Y''. The problem is, what shall we do if the committee selected by a given voting rule fails these constraints? In this paper we argue how the logic underlying weakly-separable and best-$k$ rules can be extended into an ordering of committees, and study the question of how to select the best valid committee with respect to this order. The problem is $\NP$-hard, but we show the existence of a polynomial time solution in the case of tree-like constraints, and a fixed-parameter tractable algorithm for the general case.
\end{abstract}

\section{Introduction}

Perhaps the least controversial desideratum in social choice theory is non-imposition -- the requirement that every candidate can be a winner in at least one profile. Indeed, it is hard to come up with a convincing story why an election designer should allow voters to vote for $a$, while eliminating even the theoretical possibility of $a$ winning, unless he is actively trying to provoke a revolution.

The situation changes when we consider multiwinner elections. When electing assemblies and parliaments, or selecting a movie library or a package of advertisements, we often encounter constitutional or conventional restrictions on which sets are acceptable and which are not. This could be due to equity concerns, such as the twenty-four countries around the world that reserve seats in parliament for women; protection of minority rights, such as the religious seats in Iran or the ethnic seats in Croatia; social stability, such as the Colombian peace agreement which reserved seats for former FARC combatants, or the ancient Roman requirement that one consul be a pleb; credibility, such as bipartisan committees in the United States, or the Cypriot Supreme Court which requires a Greek, Turkish, and a neutral judge; protection of culture, such as the French law requiring that forty percent of songs sung on the radio are in French; and many others.

As these are real life problems they have real life solutions, but these can lack the desirable features we seek in social choice. The parliamentary system of New Zealand uses electoral districting to offer representation based on native, regional, and party status, but the resulting system requires the size of the elected parliament to be slightly adjusted after each election. The Roman republic handled consular restrictions by sequential voting, but that gives undue power to the former of the agenda. Certain organisations \citep{Brams1990} seek to meet representation goals by simply choosing a list of candidates that will incentivise the voters to elect the right number of candidates of each type, but such a procedure offers no guarantees and is of course questionable from a democratic standpoint.

It is thus desirable to use the tools of voting theory to handle this problem, but the standard social choice model cannot handle such constraints; candidates are a list of names, distinguished only by their positions in the voters' preference orders. The function does not have access to the information that failing to elect one of $a$, $d$, or $f$ will lead to unacceptable social tension, or that a library of the highest rated films will not do if all the films are of the same genre.

In this paper we will address this problem by specifying a framework for constraints on the range of a committee selection rule, and considering how to elect the best committee out of those admitted by the constraints. In the absence of a welfare function to maximise, we seek derive what it means for a committee to be the best by considering how the logic underlying a committee selection rule can be extended to a ranking of committees from best to worst.

\subsection{Related work}

The idea that when electing a committee, not every subset of candidates has been present in the approval voting literature for a number of years. \citet{Kilgour2010} and \citet{Kilgour2012} discuss the possibility of restricting the range of an election to satisfy diversity or other constraints, and \citet{Brams1990} detail the authors' attempts to convince an organisation, interested in finding a voting system that elects a committee representative with respect to both the region the candidates come from and the specialist skills they posses, to adopt constrained approval voting.

Within the wider voting community the topic is relatively new. Closest to the present work are \citet{Bredereck2018} and \citet{Celis2018}. Both papers attempt to maximise social welfare in a setting where voters have submodular utility functions over sets of candidates (the special case of an additive utility function corresponds exactly to the weakly-separable case in this paper), but some sets are forbidden due to diversity constraints. Both papers consider what the present work calls \emph{interval constraints}, that is constraints of the form ``between 3 and 5 candidates on the committee should be a Zoroastrian''. \citeauthor{Bredereck2018} in addition to this consider arbitrary numeric constraints, which allows them to handle restrictions like ``either 2 or 7 of the candidates should be a man'', or ``the number of women should be prime''. The main focus of these works are approximation algorithms and hardness results.

Similar in spirit is the work of \citet{Yang2018}, where they seek to elect a committee satisfying some graph property. Given is a graph over the candidates, which represents a binary relation among the candidates (e.g.\ business ties), and thus the problem of electing an independent set can be interpreted as a way to minimise the potential for collusion (or electing a clique to maximise it). Among all such committees, the authors seek to find the one that maximises the score under an approval voting based voting system.

\citet{Lang2018} study the problem of multi-attribute apportionment. Given are the votes for candidates of parliament, each candidate having a number of attributes such as religion, nationality, or gender, and the task is to form a parliament of a fixed size that satisfies membership constraints, or the closest possible parliament if the constraints are unsatisfiable. Similarly, \citet{Aziz2018} also considers the case of soft constraints, where the task is not to meet the requirements exactly, but to find the closest possible committee; in this case, the author's interest is in picking a committee that is justified envy free.

\subsection{Our contribution}

We introduce a new type of constraint -- the dominance constraints -- which require that, ex post, the number of candidates elected from group $A$ is no less than the number from group $B$. We then consider the problem of electing an optimal committee subject to interval and dominance constraints, with ``optimality'' being defined with respect to a set-extension of the order underlying a best-$k$ committee selection rule. We show that this problem is 
solvable in polynomial time if the set-extension satisfies fixed-cardinality responsiveness and the dominance relation induces a tree-like ordering of the candidate labels, and NP-hard in the general case. For arbitrary label structures, we show that for weakly separable scoring rules the problem is fixed-parameter tractable in the number of labels.

\section{Preliminaries}

\subsection{Electing a committee}

An \emph{election} $E=(C,V)$ consists of a set of $m$ \emph{candidates}, $C$, and a set of $n$ \emph{voters}, $V$. Each voter $i$ is associated with a linear order over $C$, $\succ_i$, which we call $i$'s \emph{preference order}.

The problem of electing a committee is to associate an election with a set of candidates of size $k$. Thus a \emph{committee selection rule} is a function that maps an election to one or more size-$k$ committees: the set of tied winners.

A natural way to elect a candidate in the single-winner setting is via a scoring rule -- awarding a candidate a number of points based on where in a voter's preference order they appear. This idea has been extended to committee selection rules \citep{Elkind2017}, but it turns out these \emph{committee scoring rules} are much more formidable beasts than their single-winner counterparts;
the original definition ran to a page and involved maximising a function which was defined with respect to $m\times n$ auxiliary functions, each of which operated on vectors of integers. Fortunately, we do not need to define committee scoring rules in their generality as only the simplest family of them are computable in polynomial time \citep{Faliszewski2016, Faliszewski2018}, and this family is fairly straightforward.

\begin{definition}
A \emph{weakly separable committee scoring rule}, $F$, is defined with respect to a single-winner scoring function. Let $\pos_i(c)$ denote candidate $c$'s position in voter $i$'s preference order, and consider a function, $\gamma:[m]\ria\RR$, mapping positions to scores. We define the score of a candidate as:
$$\score(c)=\sum_{i\in V}\gamma(\pos_i(c)).$$
The score of a committee $X$ is the sum of the score of its members:
$$\score(X)=\sum_{c\in X}\score(c).$$
The winners under $F$ are all committees with maximal score. It is clear that such a committee will consist of $k$ candidates with the highest individual scores.
\end{definition}

\begin{example}
The following five voters wish to elect a committee of size two:
\begin{itemize}
    \item $a\succ_1c\succ_1b\succ_1 d$
    \item $a\succ_2c\succ_2b\succ_2 d$
    \item $d\succ_3c\succ_3b\succ_3 a$
    \item $d\succ_4b\succ_4c\succ_4 a$
    \item $b\succ_5c\succ_5a\succ_5d$
\end{itemize}

The most well-known scoring rule is \emph{single non-transferable vote} (SNTV). Here the single-winner scoring function is plurality: $\gamma(1)=1,\gamma(j\neq 1)=0$. The winning committee is $\set{a,d}$ with a score of 4. Under \emph{$k$-Borda} we use the Borda single-winner scoring function: $\gamma(j)=m-j$. The winning committee is $\set{b,c}$, with a score of 17. Under \emph{Bloc} the single-winner scoring function is $k$-approval ($k$ is the committee size): $\gamma(j\leq k)=1,\gamma(j>k)=0$. The committees $\set{a,c},\set{b,c}$ and $\set{c,d}$ are tied with a score of 6.
\end{example}

Weakly separable committee scoring rules lie at the intersection of two important families identified by \citet{Elkind2017}: committee scoring rules in general, and best-$k$ rules. The former do not interest us due to their complexity, but best-$k$ rules are computable in polynomial time provided the underlying ranking function is.

\begin{definition}
A \emph{best-$k$ rule}, $F$, is a committee selection rule defined with respect to single-winner ranking function $f$ (a function mapping an election to a linear order over the candidates, or a set of such orders in the case of ties).

The committee(s) selected by $F$ consist(s) of the best $k$ candidate in the order(s) produced by $f$.
\end{definition}

Weakly separable rules are precisely the best-$k$ rules where the underlying ranking function is a single-winner scoring rule.

\begin{example}
The following six voters wish to elect a committee of size two:
\begin{itemize}
    \item $a\succ_1b\succ_1d\succ_1 c$
    \item $a\succ_2b\succ_2d\succ_2 c$
    \item $a\succ_3b\succ_3d\succ_3 c$
    \item $a\succ_4c\succ_4d\succ_4 b$
    \item $b\succ_5d\succ_5a\succ_5c$
    \item $c\succ_6d\succ_6a\succ_6b$
\end{itemize}

The \emph{single-transferable vote} (STV) family of committee selection rules falls into the category of best-$k$ rules. The simplest of these involves eliminating candidates with the fewest number of first place votes until $k$ remain. The single-winner ranking function here ranks the candidates in reverse-order of elimination (see \citet{Freeman2014} for an axiomatisation). Using parallel-world tie-breaking, the single-winner ranking function produces the orders $a\succ b\succ c\succ d$ and $a\succ c\succ b\succ d$, making $\set{a,b}$ and $\set{a,c}$ the tied winning committees.

The varieties of STV used in practice typically come with a quota and a transfer rule; any candidate who has more than the quota of first-place votes is immediately put on the committee, and the excess votes they had are reallocated using the transfer rule. This is still a best-$k$ rule, but we need more care to build the single-winner order. Using the Droop quota and Gregory's transfer rule we will allocate all first-place votes in excess of 3 to the second place on a fractional basis, so $1/4$ of a vote gets added to $c$ and $3/4$ to $b$. The single-winner order is $a\succ b\succ c\succ d$, and the winning committee is $\set{a,b}$ uniquely.
\end{example}

\subsection{Acceptable committees}

Our task presupposes that not every set in the range of the committee selection rule is acceptable, due to constraints on candidate membership. We seek to elect the best committee within such constraints. To get anywhere with this, we need to address three points. First, in order to be able to discriminate among groups of candidates, we first need a way to form these candidates into groups. Second, we need to decide on what constraints on committee membership interest us and how to formalise them. Third is the thorny question of what exactly constitutes the ``best" committee.

To discriminate among candidates, we introduce a set of \emph{labels}, $\Lambda=\lambda_1,\dots,\lambda_p$. Each $\lambda_i$ is a subset of $C$. Note that depending on the interpretation of the labels we can expect a different structure of these subsets. People generally have at most one religion, so if the labels are confessional it would be reasonable to expect $\lambda_i\cap\lambda_j=\emptyset$. On the other hand, a film could have very many genres, so we should not impose any restrictions on the label relations. With the region of a candidate, we have a more interesting structure: a politician could at once represent California and Los Angeles, but not California and Texas; the label structure would be a tree. \citet{Bredereck2018} and \citet{Celis2018} consider these structures in more detail, for the purposes of this paper we are only interested in the two extremes -- labels that are disjoint, and labels without restrictions. 

In the most general sense, constraint specifications would take the form of some formal language for identifying a viable subset $S\sub 2^C$, but this generality is neither interesting nor useful. If this language is succinct, it would be computationally intractable, and if it is bloated then the input size to our algorithms would be unreasonably large. Moreover, committee restrictions typically have the limited aims of ensuring representation or balancing power, so the utility of being able to pin down arbitrary subsets is dubious. In our framework we will consider only two types of constraints:
\begin{itemize}
    \item $\Interval(p,q,\lambda_i):$ between $p$ and $q$ members of the committee from $\lambda_i$.
    \item $\Dominance(\lambda_i,\lambda_j):$ at least as many members of the committee are from $\lambda_i$ as from $\lambda_j$.
\end{itemize}
Interval constraints can be used to guarantee or limit the representation of a group in an obvious way. Dominance constraints allow us to balance the power of two or more groups, without imposing restrictions on the rest of the candidates. For example, if we wish to elect a bipartisan committee that balances Democrats and Republicans we could, if $k$ is fixed in advance, use the constraints $\Interval(k/2,k/2,\text{Dem.})$, $\Interval(k/2,k/2,\text{Rep.})$. The disadvantage of this approach is that it a priori excludes the possibility of a third party from having a spot on the committee, whereas the constraints $\Dominance(\text{Dem.}, \text{Rep.})$, $\Dominance(\text{Rep.},\text{Dem.})$ would ensure the number of Republicans and Democrats are the same, regardless of how many independent candidates are elected.

The final issue is to clarify what it means to elect the best possible committee. For weakly-separable rules the answer seems obviously enough -- find an acceptable committee with the highest score -- but with a rule like STV no scores are given, and if the winning committee violates the constraints it is not clear how to proceed.

What we need is some means to rank the committees from best to worst. In the single-winner setting ranking functions (social welfare functions; social preference functions) have been at centre-stage since inception, but the multiwinner literature to date has focused on rules which select a single, or a set of tied, winning committee(s). Indeed, a function that explicitly orders $2^C$ would be impractical; however, it is certainly possible to order committees implicitly. The most natural example of this are the committee scoring rules -- the score given to a committee provides a weak order over all of $2^C$. It is less clear what to do with a best-$k$ rule, but the interpretation of the rule as picking the best $k$ candidates independently of each other suggests an approach based on set-ranking \citep{Barbera2004}.

The idea behind set-ranking is that given an order over individual candidates, if we impose some axioms on how we wish the order on committees to behave we could reduce the space of possibilities to something more manageable. For instance, in the setting of electing the best $k$ candidates a natural criterion is responsiveness: $a\succ b$ implies $X\cup\set{a}\succ X\cup\set{b}$, whenever $a,b\notin X$. \citet{Bossert1995} has shown that responsiveness combined with a neutrality condition restricts the space of possible orders to the rank-ordered lexicographic family. These include the famous rules leximax, where we compare committees by their best elements, in case of a tie the second best element, and so on; leximin, where we compare committees by the worst, then the second-worst, etc; and less intuitive rules such as where we rank committees by the second-best element, break ties by the seventh, and continue in some arbitrary order.

Of interest to this paper is fixed-cardinality responsiveness \citep{Kurihara2017}, a strengthening of responsiveness to act on sets.

\begin{definition}
Let $\succeq$ be an order over $2^C$. We say that $\succeq$ satisfies fixed-cardinality responsiveness if, whenever $X\succeq Y$ for $|X|=|Y|=k$ and $Z\cap X,Y=\emptyset$, then $X\cup Z\succeq Y\cup Z$. 
\end{definition}

\begin{example}
The scores produced by a weakly-separable rule satisfy fixed-cardinality responsiveness. If $X$ has a larger score than $Y$, and we add the same $Z$ candidates to both committees, then we increase both score totals by the same amount, and $\score(X\cup Z)\geq\score(Y\cup Z)$. 

The leximax and leximin extensions also satisfy fixed-cardinality responsiveness. Suppose $X\succeq_{\mathit{lmax}}Y$. Let $x_i,y_i$ denote the $i$th best candidates in $X$ or $Y$, and suppose $x_i\sim y_i$ for $i<j$, and $x_j\succ y_j$. Suppose $z_1,\dots,z_k\in Z$ are the candidates better than $x$. Clearly they are also better than $y$, so when we compare the first best $k+j-1$ elements of $X\cup Z $ and $Y\cup Z$ we are comparing either $x_i$ with $y_i$, or $z_i$ with itself; in either case we are indifferent. When we reach $x_j$ in $X\cup Z $, we compare it against either $y_j$ or $z_{k+1}$, and in either case $x_j$ is better.

Other rank-ordered lexicographic extensions do not. For example, let $\succeq_2$ be the relation that compares the second-best candidate first (and the rest in an arbitrary order). If $a\succeq b\succeq c\succeq d\succeq e$ then $\set{c,d}\succeq_2\set{a,e}$, but $\set{b,c,d}\nsucceq_2\set{a,b,e}$.
\end{example}

We can now formulate the algorithmic problems of interest.

\begin{definition}
The \emph{constrained winner election} problem is the problem that takes as input an election $E$, a set of constraints $\FC$, a set of labels $\Lambda=\lambda_1,\dots,\lambda_p\sub C$, and an oracle\footnote{That is, $R$ is a routine that given two $X,Y\sub 2^C$ will output ``YES'' if $X\succeq Y$, ``NO'' otherwise.} $R$ for a weak order $\succeq\sub2^C\times 2^C$. The output is some $k$-committee $X$ that is maximal with respect to $\succeq$ out of all committees satisfying $\FC$.

The \emph{constrained winner existence} problem, in addition to the above, takes as input a reference committee $X$. The output is ``YES'' if there exists a $Y\succeq X$ that satisfies $\FC$, ``NO'' otherwise..
\end{definition}

It is easy to see that if the labels are unrestricted then simply determining whether there exists a committee satisfying $\FC$ is $\NP$-hard.

\begin{proposition}\label{prop:hardconsistency}
It is $\NP$-hard to determine whether there exists a committee $X$ satisfying a set of constraints $\FC$. Thus the constrained winner existence problem is $\NP$-complete for any onto committee selection rule.
\end{proposition}
\begin{proof}
For interval constraints we reduce from vertex cover -- given a graph $G=(\CV,\CE)$ define a candidate for every $v\in\CV$ and a label $\lambda_i$ for every edge $e_i\in\CE$, consisting of the vertices incident on $e_i$. Introduce an $\Interval(1,2,\lambda_i)$ constraint for each such label.

A similar reduction works for dominance constraints. Define a candidate for every vertex and a label for every edge as before, as well as a label for every singleton. For every edge-label $\lambda_i$, introduce the constraints $\Dominance(\lambda_i,\set{c_1})$, \dots, $\Dominance(\lambda_i,\set{c_m})$. Since at least one candidate must be elected to form a size $k$ committee, the constraints establish that at least one vertex from every edge must be chosen.
\end{proof}

It is thus necessary to restrict the problem to tractable cases. We consider two of these: where the labels are disjoint (corresponding to the 1-layered case of \cite{Bredereck2018}, and the $\Delta=1$ case of \cite{Celis2018}), and where there is only a small number of labels.

\section{Disjoint labels and knapsack}

It turns out that in the presence of dominance constraints, the case of disjoint labels is difficult even for the simplest rules.

\begin{theorem}\label{thm:hardknapsack}
The constrained winner existence problem for weakly separable voting rules with disjoint labels is $\NP$-complete, even for SNTV.

For Bloc, it is $\NP$-complete even with a constant number of voters.
\end{theorem}
\begin{proof}
Let $(G,k)$ be an instance of clique, $G=(\CV,\CE)$. Define an SNTV election with a candidate for every vertex and every edge. For every edge candidate $e_i$, define a voter that ranks $e_i$ first and the rest arbitrarily. Define a label for each singleton. For every $e_i=(v_1,v_2)\in \CE$, add the constraints $\Dominance(\set{v_1},\set{e_i})$ and $\Dominance(\set{v_2},\set{e_i})$.

Define an additional $k + k(k-1)/2$ candidates, all in a new label $\lambda_{\mathit{ref}}$, and $k(k-1)/2$ voters, each of whom ranks an arbitrary candidate in $\lambda_{\mathit{ref}}$ first, and the constraint $\Interval(0,0,\lambda_{\mathit{ref}})$.

We claim that there's a winning committee of size $k + k(k-1)/2$ that's at least as good as the committee consisting of $\lambda_{\mathit{ref}}$ (that is, a committee with score at least $k(k-1)/2$), if and only if $G$ has a clique of size $k$.

First observe that the requirement that $e_i=(v_1,v_2)$ is on a committee only if $v_1$ and $v_2$ are on the committee establishes that the committee is a subgraph -- edges cannot be present without their incident vertices. From this we can establish that no committee of size $k + k(k-1)/2$ can have more than $k(k-1)/2$ points, as that would represent a graph with more than $k(k-1)/2$ edges and less than $k$ vertices.

In order to have $k(k-1)/2$ points, then, we need to have $k$ vertices and $k(k-1)/2$ edges, and this can only be a complete graph of order $k$, that is to say a clique.

For Bloc, we first claim that clique remains hard if we restrict ourselves to the case where a clique of size $k$ contains at least half the edges of the graph, i.e. $2{k\choose 2}\geq {|\CV|\choose 2}$. To see that this is the case, given an instance of clique $(G,k)$ expand $G$ into $G'$ by adding $3|\CV|$ new vertices, adjacent to each other and to every vertex in $|\CV|$. Clearly, $G'$ contains a clique of size $k+3|\CV|$ if and only if $G$ contains a clique of size $k$, and one can verify that $2{{k+3|\CV|}\choose 2}\geq {4|\CV|\choose 2}$ for $k\geq 2$.

Consider then an instance of clique with $2{k\choose 2}\geq {|\CV|\choose 2}$. Define a candidate for every vertex, a candidate for every edge, $k+k(k-1)/2$ candidates in tribe $\lambda_{\mathit{ref}}$, and $k+k(k-1)/2$ candidates in tribe $\lambda_{\mathit{dum}}$. Define one voter that ranks $k(k-1)/2$ edges in the top positions in any order, then the other candidates. The second voter will rank the remaining edges first, then the candidates in $\lambda_{\mathit{dum}}$, then the other candidates. The third voter will rank an arbitrary $k(k-1)/2$ candidates in $\lambda_{\mathit{ref}}$ first, then the dummy candidates, then the rest.

Add the constraints $\Interval(0,0,\lambda_{\mathit{dum}})$ and $\Interval(0,0,\lambda_{\mathit{ref}})$, and for every edge $e=(v_1,v_2)$, $\Dominance(\set{v_1},\set{e})$, $\Dominance(\set{v_2},\set{e})$. From hereon replicate the argument for SNTV.
\end{proof}

While this is an inauspicious beginning, the similarity of the problem to partial-order knapsack \citep{Johnson1983} suggests a course of attack. Partial-order knapsack is the problem of solving a knapsack instance where a partial-order over the items specifies in what order the items can be taken; this problem is strongly $\NP$-complete in the general case, but a pseudopolynomial solution exists when the partial-order is a tree. If we limit the admissible dominance constraints in a similar way, we can obtain a fully polynomial solution to the constrained election problem.

\begin{definition}
The \emph{dominance graph} associated with a constrained election instance $(E,\FC,\Lambda,R)$ is the transitive closure of the graph which has a vertex for every label, and a directed edge $(\lambda_x,\lambda_y)$ for every $\Dominance(\lambda_x,\lambda_y)$ constraint.

A dominance graph is \emph{tree-like}, if there exist no vertices $x,y,z$ such that there exist edges $(x, z)$ and $(y, z)$, but there is no edge between $x$ and $y$.
\end{definition}
The dominance graph captures all the information imparted by the dominance constraints -- whenever $(\lambda_x,\lambda_y)$ is an edge then either a direct constraint, or a consequence of several constraints, requires us to take at least as many candidates from $\lambda_x$ as from $\lambda_y$.

We can now formulate our main result.

\begin{theorem}\label{thm:easyknapsack}
Consider a constrained winner election problem with the following properties:
\begin{enumerate}
    \item The order over committees, $\succeq$, satisfies fixed-cardinality responsiveness.
    \item $\succeq$ is computable in polynomial time.
    \item $\Lambda$ consists of disjoint labels.
    \item The dominance graph induced by the constraints is tree-like.
\end{enumerate}
The problem is solvable in polynomial time.
\end{theorem}

We will solve the problem in two steps. First we will construct an algorithm that works with dominance constraints only, then we will show how to use this algorithm as a subroutine to solve the case with interval constraints.

\begin{proof}[dominance constraints only]
As a preprocessing step we will construct the dominance graph and identify every maximal clique in it. Taking the transitive closure of the dominance constraints can be done in polynomial time, e.g.\ by depth-first search, and in a transitive graph clique detection reduces to cycle detection. For convenience, we will use the language of trees -- we will refer to these maximal cliques as nodes; nodes which have no outgoing edges to other nodes as leaves; the node with no incoming edges as the root; and whenever $x$ has an outgoing edge to $y$, and there is no $z$ such that $x$ has an edge to $z$, and $z$ to $y$, we say $y$ is the child of $x$.

The algorithm is based on dynamic programming. Every node in the dominance-graph, $x$, is associated with the table $M_x[\cdot,\cdot]$ which maps two integer parameters to a committee -- $M_x[k',q]$ contains the best size $k'$ committee that can be built from the subgraph rooted at $x$, at most $q$ members of which come from the root of $x$. Thus the optimal committee will be the entry $M_r[k,k]$, where $r$ is the root of the entire graph.

In the base case we consider the leaf nodes of the graph. Such a node $x$ is a clique of $p$ labels, and its table is filled as follows:
\[
M_x[k',q]=\begin{cases}
\text{The best $r$ candidates from each of the $p$ labels,}&k'=rp,q\geq k'\\
\mathtt{null},&\text{else.}
\end{cases}
\]
In this case the optimal committee consists of the best $r$ candidates from every label. If $q<k'$, or $k'\neq rp$, no such committee is possible, and we mark this with a special symbol $\mathtt{null}$; for the inductive cases, the union of $\mathtt{null}$ with any committee evaluates to $\mathtt{null}$. 

In the general case we have a node $x$ with $t$ children. Here we have the auxiliary tables $R_x$ and $C_x$, $R[k',q]$ containing the best size $k'$ committee taken from $x$ alone, $C[p,k',q]$ the best size $k'$ committee taken from the leftmost $p$ children of $x$.

The table $R_x$ is filled in the same way as in the base case.
\[
R_x[k',q]=\begin{cases}
\text{The best $r$ candidates from each of the $p$ labels,}&k'=rp,q\geq k'\\
\mathtt{null},&\text{else.}
\end{cases}
\]

The entries $C[1,k',q]$ contain the best committee from the leftmost child, $y$, alone:
\[
C_x[1,k',q]=M_y[k',q].
\]
To pick the best size $k'$ committee from the leftmost $p+1$ children, we have to take $r$ candidates from the $(p+1)$th child, $y$, and $k'-r$ children from the remaining $p$:
\[
C_x[p+1,k',q]=\max_{s\leq k'}(M_y[s,q]\cup C_x[p,k'-s,q]).
\]
The $\max$ operator means to pick the best committee with respect to $\succeq$.

The primary table of node $x$ is then the best way to take $s$ candidates from the root and $k'-s$ from the children, with the proviso that we cannot take more than $s$ candidates from the root nodes of the children.
$$M_x[k',q]=\max_{s\leq \min(k',q)}(R[s,q]\cup C[t,k'-s,s]).$$

This completes the construction. Let us verify that it can be done in polynomial time.

The base case tables of $M_x$ for leaf nodes and $R_x$ for inner nodes contain $k^2$ entries, and to fill an entry all we need to know is the best $k$ candidates -- and we can recover $\succeq$ over singletons with $m^2$ queries.

In a node with $t$ children, $t\leq m$ so the table $C_x$ has no more than $mk^2$ entries. The entries $C_x[1,k',q]$ are filled in constant time, the case $C_x[p+1,k',q]$ requires us to find the best of $k'\leq k$ possible committees, which can be done with at most $k$ queries of $\succeq$.

The inductive case tables $M_x$ have $k^2$ entries, and to fill each we need at most $k$ queries of $\succeq$.

Now let us verify that the construction is correct. The base case tables of $M_x$ for leaf nodes and $R_x$ for inner nodes are correct from fixed-cardinality responsiveness. Likewise for the entries $C_x[1,k',q]$.

For entries $C_x[p+1,k',q]$, let $X$ be the best size $k'$ committee taken from the leftmost $p+1$ children, and suppose it contains $s$ candidates from the $(p+1)$th child. Let $X_1$ be those candidates on the committee from the leftmost $p$ children, and $X_2$ from the $(p+1)$th.

Let $Y$ be the best size $k'-s$ committee taken from the leftmost $p$ children and $Z$ the best size $s$ committee from the $(p+1)$th child. By fixed-cardinality responsiveness, $X_1\cup Z\succeq X_1\cup X_2$, so $X\sim X_1\cup Z$, and $Y\cup Z\succeq X_1\cup Z$, so $X\sim Y\cup Z$. That is, the best size $k'$ committee consists of the best size $k'-s$ committee from the first $p$ children and the best size $s$ committee from the $(p+1)$th child. By induction, these are precisely the committees in $C_x[1,k'-s,q]$ and $M_y[s,q]$.

Likewise for the inductive case $M_x$ tables.
\end{proof}

The final step is managing the interval constraints. The upper bounds imposed by these constraints will be handled by throwing away the excess candidates, the lower bounds by modifying the order over the committees.

\begin{proof}[interval constraints]
Consider a constrained winner election problem instance $P=(E,\FC, \Lambda, R)$ where the dominance graph is tree-like. We will transform it into an instance $P^*=(E^*,\FC^*, \Lambda^*, R^*)$ without interval constraints, with the property that if there exists a solution to $P$, then the solution to $P^*$ is also a solution to $P$. We can then solve $P$ in polynomial time by solving $P^*$, and verifying whether the solution to $P^*$ satisfies the constraints of $P$.

We will call a constraint $\Interval(p,q,\lambda_x)$, with $q<|\lambda_x|$ an \emph{upper-bound}. Our first step will be to remove all upper bounds from $P$. Observe that if we cannot take more than $q$ candidates from $\lambda_x$, then it follows that we cannot take more than $q$ candidates from any descendant of $\lambda_x$ in the dominance graph. We will thus begin by propagating these constraints through the graph: for every upper bound $\Interval(p,q,\lambda_x)$, and every descendant of $\lambda_x$, $\lambda_y$, add the constraint $\Interval(0,q,\lambda_y)$. Next, for every label $\lambda_x$, let $\Interval(p,q,\lambda_x)$ be the upper bound for $\lambda_x$ with the smallest $q$. Remove all but the best $q$ candidates from $\lambda_x$. Observe that every constraint of the form $\Interval(0,q,\lambda_x)$ is now redundant and can be removed. The resulting problem, $P'=(E',\FC', \Lambda', R)$, will have no upper-bounds, and by fixed-cardinality responsiveness the solution to $P'$ must also be a solution to $P$.

Next, consider constraints $\Interval(p,q,\lambda_x)$, $p>0$. Clearly, if we are to include at least $p$ elements from $\lambda_x$ on the committee, we must include at least $p$ from every ancestor of $\lambda_x$ in the dominance graph. Thus we propagate the constraints $\Interval(p,m,\lambda_x)$ upwards.
Next, for every label $\lambda_x$ let $\Interval(p,q,\lambda_x)$ be the constraint with the largest $p$; call this $p$ the \emph{lower-bound} for $\lambda_x$. Observe that if the committee has at least the lower-bound of candidates from $\lambda_x$, it will satisfy all the interval constraints involving $\lambda_x$.

Call the best lower-bound candidates from every $\lambda_x$ the \emph{obligatory} candidates. Let $R^*$ function as follows: if $X$ and $Y$ both contain the same number of obligatory candidates, then $X\succeq_{R^*}Y$ whenever $X\succeq_R Y$. Otherwise, $X\succeq_{R^*}Y$ if and only if $X$ contains at least as many obligatory candidates as $Y$. Clearly $R^*$ is computable in polynomial time whenever $R$ is. Observe further that $R^*$ satisfies fixed-cardinality responsiveness -- consider $X\succeq_{R^*}Y$, and some $Z$ disjoint with $X,Y$. If $X$ has more obligatory candidates than $Y$, then $X\cup Z$ will have more than $Y\cup Z$; if $X\succeq_{R^*} Y$ then $X\cup Z\succeq_{R^*}Y\cup Z$ by the definition of $R^*$. If $X$ had the same number of obligatory candidates as $Y$, then it must mean $X\succeq_R Y$, and by fixed-cardinality responsiveness of $R$, $X\cup Z\succeq_{R}Y\cup Z$. Since $X\cup Z$ and $Y\cup Z$ will have the same number of obligatory candidates, it follows that $X\cup Z\succeq_{R^*}Y\cup Z$.

Now let $P^*=(E',\FC^*,\Lambda',R^*)$, where $E'$ is the initial candidates minus the ones removed for the upper-bounds, $\FC^*$ consists of the dominance constraints of $\FC$, $\Lambda'$ is the restriction of $\Lambda$ to the remaining candidates in $E'$, and $R^*$ is the order just described. $P^*$ satisfies all the conditions of \autoref{thm:easyknapsack}, so we can find a solution, $X$, in polynomial time. If this solution satisfies $\FC$ but is not a solution to the original problem, it must follow that the solution to the original either contains deleted candidates or does not contain all the obligatory candidates -- but that contradicts fixed-cardinality responsiveness.

If $X$ does not satisfy the constraints in $\FC$, then this must mean that $X$ does not contain all the obligatory candidates. Since $R^*$ would prioritise any committee that has these candidates over those that do not, it follows that no committee which contains all the obligatory candidates can satisfy the dominance constraints, hence the original problem has no solution.
\end{proof}

\section{Small number of labels and fixed-parameter tractability}

If the number of label is low we can obtain a polynomial time solution for the case of weakly-separable rules via mixed integer linear programming, using similar techniques to the result for top-$k$ counting rules by \cite{Faliszewski2018} and Theorem 10 of \cite{Bredereck2018}.

\begin{theorem}
The constrained winner election problem for weakly separable rules is fixed parameter tractable with respect to the number of labels.
\end{theorem}
\begin{proof}
Intuitively, if the number of labels is constant this problem can be solved by brute force. Observe that if $X$ is a committee satisfying constraints $\FC$, then if $x\in X$, $y\notin Y$, $\score(y)>\score(x)$ and for all labels $\lambda_i$, $x\in \lambda_i$ if and only if $y\in \lambda_i$, then $Y=(X\setminus\set{x})\cup\set{y}$ is also a committee satisfying the constraints, and $\score(Y)\geq\score(X)$. In other words if two candidates have the same labels, we will never violate a constraint by swapping a low scoring one for a high scoring one, and doing so will only increase the score.

Since a constant $p$ labels gives rise to a Venn diagram with a constant $q$ regions, we need only consider all ways of choosing the best $k$ candidates from every region, of which there are at most $k^q$.

To obtain a fixed parameter tractable algorithm, we will recast the intuition above as a mixed integer linear program. Such a program is fixed parameter tractable in $p$ if the number of integer variables is a function of $p$ alone \citep{Lenstra1983}.

Let $D_1,\dots,D_q$ enumerate the regions of the Venn diagram induced by $\Lambda$. Note that $q\leq 2^p$. Introduce the integer variables $\mathtt{d_1},\dots,\mathtt{d_q}$, the interpretation of $\mathtt{d_i}$ being the number of elements taken from $D_i$. Introduce the real indicator variables $(s_{i,j})_{i\leq q,j\leq k}$ with the interpretation that $s_{i,j}=1$ if and only if at least $j$ elements are taken from $D_i$..

The constant values $c_{i,j}$ will represent the score of the $j$th highest scoring candidate from $D_i$. The resulting system is in \autoref{fig:milp}.

\begin{figure}
\begin{align*}
\max\quad&\sum_{i\leq q}\sum_{j\leq k} s_{i,j}c_{i,j}&\\
\text{s.t.}\quad&&\\
    1)\quad&\sum_{i\leq q} \mathtt{d_i}=k,&\\
    2)\quad&\sum_{D_j\sub \lambda_i}\mathtt{d_j}\geq r,&\\
    \quad&\sum_{D_j\sub \lambda_i}\mathtt{d_j}\leq t,&\text{for all }\Interval(r,t,\lambda_i)\\
    3)\quad&\sum_{D_j\sub \lambda_i}\mathtt{d_j}\geq\sum_{D_j\sub \lambda_{t}}\mathtt{d_j},&\text{for all }\Dominance(\lambda_i,\lambda_{r})\\
    4)\quad&\sum_{i\leq q}\sum_{j\leq k}s_{i,j}=\mathtt{d_i},&\\
    5)\quad&0\leq s_{i,j}\leq 1,&i\leq q,j\leq k
\end{align*}
\caption{A mixed integer formulation of the constrained winner problem. Integer variables in typewriter font.}
\label{fig:milp}
\end{figure}

Constraint 1 ensures the committee is of size $k$, constraints 2 ans 3 ensure the satisfaction of $\FC$, and constraint 4 establishes the relation between the $\mathtt{d_i}$ variables and the objective function. Clearly if the system arrives at a solution where all $s_{i,j}$ are integral, we have found a maximal score $k$-committee.

Suppose then that the solution is such that there exists a $0<s_{i,j}<1$. Since $\sum_{j\leq k}s_{i,j}=\mathtt{d_i}$, and $\mathtt{d_i}$ is an integer, it follows there exists another $0<s_{i,j'}<1$. Without loss of generality, let $j'>j$. Since $c_{i,j}\geq c_{i,j'}$ by definition, the value of the solution will not decrease if we transfer the weight from $s_{i,j'}$ to $s_{i,j}$, after which we will have reduced the number of non-integral values by one. By repeating this process for each non-integral $s_{i,j}$, we will find an integral solution with the same value.
\end{proof}

\section{Conclusion}

We have introduced a framework for specifying constraints on committee membership, with the novelty of dominance constraints, and formulated the constrained election problem via set-extensions of the ranking function of a best-$k$ rule. This problem proved to be $\NP$-hard in the general case, but if the dominance constraints imposed a tree-like structure on the labels, a polynomial solution via dynamic programming is possible. We also demonstrated fixed-parameter tractability of the problem in the number of labels for weakly separable rules.

The study of constrained elections is new, and many questions remain open, but we feel a particularly interesting future direction would be to consider where these constraints come from. To date, works on the subject have assumed that constraints are given exogenously, and the tools of social choice are limited to finding a solution that fits; whereas presumably, constraints are set with a goal in mind. It would be interesting to integrate this goal into the problem itself.

\bibliographystyle{apalike}
\bibliography{references.bib}   %

\end{document}